\documentclass[11pt]{article}

\usepackage[a4paper]{geometry}
\usepackage{comment}

\usepackage{url}
\usepackage{amsmath}
\usepackage{amssymb}
\usepackage{wrapfig}
\usepackage{graphicx}
\graphicspath{{./Figs/}}
\usepackage{theorem}
\usepackage{citesort}
\usepackage{algorithmic}
\usepackage[ruled]{algorithm}

\newenvironment{proof sketch}[1]{\noindent {\emph{Proof sketch of #1:}}}{\hfill
 \qed}
\newtheorem{theorem}{Theorem}

\newtheorem{lemma}{Lemma}
\newtheorem{fact}{Fact}
\newtheorem{corollary}{Corollary}

\newtheorem{proof}{Proof}

\newcounter{Codeline}
\newcommand{\Newcodeline}{\setcounter{Codeline}{1}}
\newcommand{\Cl}{{\theCodeline}: \addtocounter{Codeline}{1}}
\newcommand{\crm}{\\}

\newcommand{\ot}{\ensuremath{\leftarrow}}

\newcommand{\qed}{\hfill $\Box$}

\newcommand{\Scomment}[1]{\ensuremath{/\ast} #1 \ensuremath{\ast/}}

\allowdisplaybreaks

\begin{document}

\title{A New Direction for Counting Perfect Matchings}

\author{Taisuke Izumi\thanks{Nagoya Institute of Technology, Gokiso-cho,
Showa-ku, Nagoya, Aichi, 466-8555, Japan (Corresponding Author).} 
\and Tadashi Wadayama\thanks{Nagoya Institute of Technology} }

\maketitle \thispagestyle{empty}

%=========================================================================
%  Abstract
%=========================================================================

\begin{abstract}
In this paper, we present a new exact algorithm for counting 
perfect matchings, which relies on neither inclusion-exclusion 
principle nor tree-decompositions. For any bipartite graph of $2n$ nodes 
and $\Delta n$ edges such that $\Delta \geq 3$, our algorithm 
runs with $O^{\ast}(2^{(1 - 1/O(\Delta \log \Delta))n})$ time 
and exponential space. Compared to the previous algorithms, it achieves 
a better time bound in the sense that the performance degradation to 
the increase of $\Delta$ is quite slower. The main idea of our algorithm 
is a new reduction to the problem of computing the cut-weight distribution 
of the input graph. The primary ingredient of this reduction is
MacWilliams Identity derived from elementary coding theory. The whole 
of our algorithm is designed by combining 
that reduction with a non-trivial fast algorithm computing 
the cut-weight distribution. To the best of our knowledge, the approach 
posed in this paper is new and may be of independent interest.
\end{abstract}

\section{Introduction}

%The {\em permanent} of a $n \times n$ matrix $A = (A_{ij}) \in 
%\mathbb{R}^{n \times n}$ is the value defined as follows:
%\[
%\mathit{perm}(A) = \sum_{\sigma \in S_n} \prod^n_{i=1} A_{i\sigma(i)},
%]
%where $S_n$ is the set of all permutations on $[1, n]$.

Counting perfect matchings in given input graph $G$ is recognized 
as one of hard combinatorial problems. In particular, the case that
$G$ is bipartite has attracted much 
attention with its long history because of the relation to the 
computation of permanent, which is a characteristic value of 
matrices with many important applications. 
Since counting perfect matchings for bipartite graphs belongs to 
\#P-complete, there seems to be no algorithm which runs within 
polynomial time for any input. Thus all of the previous studies lies on
one (or more) of the following directions: Approximation, restriction of
input graphs, or exact exponential algorithms. In this paper, we focus
on the third line.

A seminal exponential-time algorithm for counting perfect
matchings is Ryser's one based on the inclusion-exclusion 
principle~\cite{Ryser63}. For any bipartite graph $G$ of 
$2n$ vertices, it counts perfect matchings with $O^{\ast}(2^n)$ 
time\footnote{$O^{\ast}$ means the Big-O notation with omitting 
$\mathrm{poly}(n)$ factors.} and polynomial memory space. There
has been several improvements following that work:
Bax and Franklin have shown an algorithm running with 
$O^{\ast}(2^{(1 - 1/O(n^{2/3}\ln n))n})$ expected time and 
exponential space~\cite{BF02}. Servedio and Wan have given an 
algorithm with a time upper bound depending on the average degree 
$\Delta$~\cite{SW05}. It achieves $O^{\ast}(2^{(1 - 1/O(\exp(\Delta)))n})
$ time and polynomial space. Another approach to this problem is the
usage of tree decompositions~\cite{ALS91,RBR09}. By combining the fact
that sparse graphs have a treewidth less than $(1 - \epsilon) n$ 
for some constant $\epsilon$ (e.g., if $\Delta \leq 3$, $\epsilon \approx 5/6$
holds~\cite{FK06}), we can obtain an algorithm running 
$O^{\ast}(2^{(1- \epsilon) n})$ time.  All of these algorithms 
break $O^{\ast}(2^n)$-time barrier in some sense. However, 
during last 50 years, there has been proposed no algorithm 
achieving exponential-time speedup for {\em any} graph, which 
is a big open problem in this topic.

Our result presented in this paper can be put on the same line.
The main contribution is to propose a new algorithm for counting
perfect matchings. For any bipartite graph of $2n$ nodes and 
$\Delta n$ edges, it runs with $O^{\ast}(2^{(1 - 1/O(\Delta 
\log \Delta))n})$ time and exponential space. While this algorithm 
does not settle the open problem stated above, its speed-up factor 
becomes substantially closer to the exponential compared to 
the previous algorithms. 

An important remark is that the approach we adopt is quite different 
from any previous solutions. It relies on neither 
inclusion-exclusion nor tree decomposition. Actually, the main
idea is an extremely-simple reduction to the problem of computing the 
cut-weight distribution of the input graph. The precise construction 
of our algorithm can be summarized as follows:
\begin{itemize}
\item For any {\em odd} input bipartite graph $G$ of $2n$ nodes and 
$m$ edges, we can show that the number of $G$'s perfect matchings is 
equal to the number of elements with weight $m - n$ in its cycle space.
In addition, for any bipartite graph $G$, it is possible to construct the
odd bipartite graph $\tilde{G}$ which has the same number of perfect
matchings as $G$, by adding a constant number of nodes. 
\item By utilizing the primal-dual relation between cycle space
and cut space, we can reduce the problem of counting cycle-space elements 
with weight $m-n$ to computing the weight distribution of the cut space.
The technical tool behind this reduction is the use of
MacWilliams identity, which is a well-known theorem derived from 
elementary coding-theory. That identity provides the linear transformation
(by so-called {\em Krawtchouk matrices}) that maps the 
weight-distribution vector of any cut space to the corresponding 
cycle space. 
\item Since the cardinality of the cut space is vertex-exponential, 
it is easy to construct a naive algorithm with $O^{\ast}(2^{2n})$ running 
time. We improve its running time by utilizing the bipartiteness property 
and a novel technique analogous to separator decompositions. 
\end{itemize}

It should be noted that except for the last step, our approach 
is applicable to any graphs which may not be bipartite. Our reduction
technique can be seen as an algebraic approach to the design of 
exact algorithms as considered in~\cite{BHKK07,LN10}, where several kinds of 
algebraic transformations are used for appropriate handling of 
target universes. To the best of our knowledge, this is the first 
attempt using the transformation by MacWilliams Identity (or 
equivalently Kratwtchouk matrices) for that objective.

The organization of the paper is as follows: 
We first presents several notions and definitions in Section~\ref{sec:CodingTheory}, which includes an tiny tutorial of linear
codes. Section~\ref{sec:Reduction} introduces our reduction to cut space. 
The algorithm to compute the cut-weight distribution is shown in 
Section~\ref{sec:Algorithm} gives
an algorithm computing the cut space. We mention the
related work in Section~\ref{sec:RelatedWork}, and finally conclude
the paper in Section~\ref{sec:Conclusion} with the open problems
posed by our result.

\section{Preliminaries from Coding Theory}
\label{sec:CodingTheory}

A linear code $C$ over $\mathbb{F}_2$ defined by $n \times m$ 
matrix $M$ is the set of $m$-dimensional vectors as follows: 
\[
C = \{\boldsymbol{v}M | \boldsymbol{v} \in \mathbb{F}^n_2\}.
\]
The matrix $M$ is called the {\em generator matrix} of $C$. By the
definition, code $C$ is the linear subspace of $\mathbb{F}_{2}^m$ 
spanned by the row vectors of $M$. The rank of that subspace
is denoted by $\mathit{rank}(M)$. Clearly, the number of codewords
in $C$ (denoted by $|C|)$ is equal to $2^{\mathit{rank}(M)}$. 
A $(m ,r)$-linear code is the one such that the length of codewords 
is $m$ and its rank is $r$. 

Let $C$ be a linear code with generator matrix $M$. The {\em parity
check matrix} $H$ of $C$ is the $m \times (m - \mathrm{rank}(M))$ 
matrix satisfying $H\boldsymbol{w}^T = \boldsymbol{0}$ for any 
codeword $\boldsymbol{w} \in C$. It is well-known that there is
a duality between generator matrices and parity check matrices:
For the code $C^\perp$ with generator matrix $H$, it is easily verified
that $\boldsymbol{v}^{T} M = 0$ holds for any $\boldsymbol{v} \in C^\perp$. 
That is, $M$ is the parity check matrix of $C^\perp$. Then the code 
$C^\perp$ is called the {\em dual} code of $C$. Obviously
$\boldsymbol{v}^T \boldsymbol{v}^\perp = 0$ holds for any 
$\boldsymbol{v} \in C$ and 
$\boldsymbol{v}^{\perp} \in C^\perp$. It implies that the dual code is the orthogonal 
complement of the primary code.

Given a codeword $\boldsymbol{w}$, the number of appearance of value $1$
in $\boldsymbol{w}$ is called the {\em weight} of $\boldsymbol{w}$.
The {\em weight distribution} of a $(m, r)$-linear code $C$ is 
the $m$-dimensional vector whose $k$-th entry $W_C[k]$ is 
the number of codewords with weight $k$ in $C$. The weight
distribution is often represented as the form of generating functions
$F_C(x) = \sum_{w = 0}^{m} W_C[w]x^w$. This function is called
the {\em weight-distribution polynomial} of $C$. There is a well-known 
theorem providing a relationship between
the weight-distribution polynomials of primary and dual codes:

\begin{theorem}[MacWilliams Identity~\cite{MS77}] \label{thmMcWilliams}
Let $C$ be a $(m, r)$-linear code over $\mathbb{F}_2$ and 
$C^\perp$ be its dual. Then, the following identity holds:
\[
F_{C}(x) = \frac{1}{2^r}(1 + x)^m
F_{C^{\perp}}\left(\frac{1-x}{1+x}\right).
\]
\end{theorem}

By comparing the coefficient of each monomial in both sides, we have 
the representation of $W_C[k]$ by a linear sum of the weight distribution 
of $C^\perp$:
\begin{equation} \label{math:McWilliams}
W_{C}[i] = \frac{1}{2^r} \sum_{j=0}^{m} K_m(j, i) W_{C^{\perp}}[j],
\end{equation}
where $K_m(j, i)$ is the value known as Krawtchouk polynomials, defined
as follows:
\[
K_m(j,i) = \sum_{k=0}^{m} (-1)^k {i \choose k} {{m - i}\choose {j - k}}.
\]

\section{Counting Perfect Matchings via Cycle Space}
\label{sec:Reduction}

\subsection{Cut and Cycle Spaces}

In this section any arithmetic operation for elements of vectors and matrices
is over field $\mathbb{F}_2$. Letting $G = (V, E)$ be an undirected graph
with $n$ vertices $v_1, v_2, \cdots, v_n$ and $m$ edges $e_1, e_2, \cdots, e_m$, 
its {\em incidence matrix} $A^G = (A^G_{i,j}) \in \mathbb{F}_2^{n \times m}$ 
is the one such that $A^{G}_{i,j} = 1$ if and only if $v_i$ is 
incident to $e_j$ and $A^{G}_{i,j} = 0$ otherwise. It is easy to check that 
the $i$-th row of $A^G$ is the 0-1 vector representation of the set of 
edges incident to $v_i$. Given a 0-1 (row) vector representation 
of $\boldsymbol{v}_S$ for a vertex subset 
$S \subseteq V$, $\boldsymbol{v}_S A^G$ is the cutset between $S$ and 
$V \setminus S$. It implies that the linear code defined by the generator 
matrix $A^G$ is equivalent to the family of edge subsets each of 
which represents a cutset, so-called the {\em cut space} of $G$.

As an well-known fact, the set of all cycles in $G$ induces a linear 
subspace of $\mathbb{F}_2^m$, where each element is a 0-1 vector 
representation of the edge set constituting one or more cycle(s). This 
subspace is called the {\em cycle space} of $G$. Note that the 
cycle space can be recognized as the set of all spanning even subgraphs 
(i.e., subgraphs where every vertex has an even degree). The matrix whose 
row is the basis of $G$'s cycle space is denoted by $B^G$. Similarly
to the cut space, we regard the cycle space as a linear code defined
by the generator matrix $B^G$. An important relationship between 
cut space and cycle space, stated below, is known: 

\begin{fact}
The cycle space of $G$ is the orthogonal complement of the cut space
of $G$.
\end{fact}

This fact implies that the linear code associated with a cycle space
is the dual code of that with the corresponding cut space, and vice versa. 
In the following argument, given an undirected graph $G$, $C(G)$ and 
$C^\perp(G)$ denote the linear codes defined by the generator matrices 
$B^G$ and $A^G$ respectively. We often use term
``cutset of $G$'' as the meaning of the codeword of $C(G)$
associated with that cutset. The same usage is also applied for cycle spaces.

\subsection{From Cycle Space to Number of Perfect Matchings }

Given an undirected graph $G = (V, E)$, we consider 
counting the number of perfect matchings of $G$. Since there is 
no perfect matching if the number of vertices is odd, we define 
$2n = |V|$. Let $m = |E|$ for short.  The degree
of vertex $v$ is denoted by $d(v)$. First we focus on the case that 
$G$ is an {\em odd} graph, i.e., a graph such that $d(v)$ is odd for any
$v$ in $V$. The number of perfect matchings of 
odd graph $G$ is related to $G$'s cycle space by the following lemma.

\begin{lemma} \label{lmaMatchingCode}
For any odd graph $G$, the number of perfect matchings in $G$
is equal to $W_{C(G)}[m - n]$.
\end{lemma}

\begin{proof} \normalfont 
Let $V = \{v_0, v_1, \cdots, v_{2n-1}\}$ be the set of vertices in $G$.
We prove the lemma by defining a bijection between the set of 
codewords with weight $m - n$ and perfect matchings. More precisely,
we prove that the complement (in terms of the edge set of $G$) of any 
codeword $\boldsymbol{w}$ in $C(G)$ with weight $m - n$ is a 
1-factor (equivalent to 
a perfect matching). Let $G_w$ be a spanning even subgraph 
corresponding to $w$. The degree of $v_i \in V$ in $G_w$ is 
denoted by $d'(v_i)$. To prove that the complement of $G_w$ is a 
1-factor, it suffices to show that $d'(v_i) = d(v_i) - 1$ holds for any 
$v_i \in V$. Suppose for contradiction that $d'(v_i) \neq d(v_i) - 1$ holds 
for some  $v_i \in V$. Since $d'(v_i) \leq d(v_i)$, $d(v_i)$ is odd, and 
$d'(v_i)$ is even (recall $G_w$ is a spanning even subgraph of $G$), 
we have $d'(v_i) < d(v_i) - 1$. To make 
$\sum^{n-1}_{i = 0} d'(v_i)= 2(m - n)$ hold, there must exist another vertex
$v_j$ satisfying $d'(v_j) > d(v_j) -1 \Rightarrow d(v_j) = d'(v_j)$. 
It contradicts the fact that $d(v_j)$ is odd. 
\end{proof}

Combining the lemma above and Theorem~\ref{thmMcWilliams},
we obtain the following corollary:

\begin{corollary}
Let $G$ be an arbitrary odd graph. There exists an algorithm to 
count the number of perfect matchings in $G$ with $O(m\tau(5m))$ time
provided that the weight distribution $W_{C^\perp(G)}$ is available,
where $m$ is the number of edges in $G$ and $\tau(x)$ be the time 
required for arithmetic operations of two $x$-bit integers.
\end{corollary}

Note that the absolute value of Krawtchouk polynomials has a trivial 
upper bound $|K_m(j, i)| \leq \mathrm{poly}(m){m \choose m/2}^2 \leq 2^{2m +
O(\log m)}$, and the number of all codewords of $C^\perp(G)$ is at most 
$2^n \leq 2^m$. Thus, the time required for each arithmetic operation 
in the right term of formula~\ref{math:McWilliams} is bounded by $\tau(5m)$.

\subsection{Transformation to Odd Bipartite Graph}

While the result in the previous subsection assumes that $G$ is an odd 
graph, that assumption can be easily removed. The fundamental idea is 
to construct the odd graph $\tilde{G}$ that has the same number of 
perfect matchings as $G$. While we only consider the case that
$G$ is a bipartite graph in this paper, general graph can be
handled similarly. Let $G = (V_1 \cup V_2, E)$, 
be an arbitrary bipartite graph such that $|V_1| = |V_2| = n$, and 
$V = V_1 \cup V_2$ for short. The set of even-degree vertices in $V_i$ is
denoted by $V^{\mathrm{even}}_i$ ($i \in \{1, 2\}$). We can easily show 
the following lemma:

\begin{lemma} \label{lmaEvenOdd}
The values of $|V^{\mathrm{even}}_1|$ and $|V^{\mathrm{even}}_2|$ have
the same parity.
\end{lemma}

\begin{proof} \normalfont
Assume that $|E|$ is odd. Since 
$\sum_{v \in V^{\mathrm{even}}_i} d(v)$ is even for any $i \in \{1, 2\}$,
$\sum_{v \in V \setminus V^{\mathrm{even}}_i} d(v)$ must be odd.
Thus, $|V \setminus V^{\mathrm{even}}_i|$ is odd for any $i \in \{1, 2\}$
because any node in $V \setminus V^{\mathrm{even}}_i$ has an odd degree.
It implies that $|V^{\mathrm{even}}_i|$ is odd for any $\{1, 2\}$. The 
case of even $|E|$ can be proved similarly. 
\end{proof}

The construction of $\tilde{G}$ is given as follows:

\begin{itemize}
\item Add two vertices $\tilde{v}_{i, 1}$, and $\tilde{v}_{i, 2}$ 
to $V_i$ for each $i \in \{1, 2\}$. 
%That is, $\tilde{V}_i = V \cup \{\hat{v}_{i, 1}, \tilde{v}_{i, 2}\}$ $(i \in \{1, 2\})$.
\item For each $i \in \{1, 2\}$, connect each node in 
$V^{\mathrm{even}}_i$ with $\tilde{v}_{3 - i, 1}$, and $\tilde{v}_{3 - i, 1}$
with $\hat{v}_{i, 2}$. 
\item If $d(\tilde{v}_{1, 1})$ and $d(\tilde{v}_{2, 1})$ are
even, connect them. Recall that $d(\tilde{v}_{1, 1})$ and $d(\tilde{v}_{2, 1})$
have the same parity  from Lemma~\ref{lmaEvenOdd}.
\end{itemize}

An example of the construction is shown in Figure~\ref{fig:bipartiteConstruction}.
For the constructed graph $\tilde{G}$, we have the following lemma. 

\begin{lemma} \label{lmaBiparttiteGraph} 
The graph $\tilde{G}$ is an odd bipartite graph, and has the same 
number of perfect matchings as $G$.
\end{lemma}

\begin{proof} \normalfont
Any node in $\tilde{G}$ clearly has an odd degree. Let $M \subseteq \tilde{E}$ 
be any perfect matching of $\tilde{G}$. Since $\tilde{v}_{1, 2}$ and
$\tilde{v}_{2, 2}$ is degree one, edges $\{\tilde{v}_{1, 1}, \tilde{v}_{2, 2}\}$
and $\{\tilde{v}_{2, 1}, \tilde{v}_{1, 2}\}$ are necessarily included in $M$.
Then $M \setminus \{\{\tilde{v}_{1, 1}, \tilde{v}_{2, 2}\},
\{\tilde{v}_{2, 1}, \tilde{v}_{1, 2}\}\}$ is a perfect matching of $G$.
Conversely, given a perfect matching $M' \subseteq E$ of $G$, $G \cup \{\{\tilde{v}_{1, 1}, \tilde{v}_{2, 2}\}, 
\{\tilde{v}_{2, 1}, \tilde{v}_{1, 2}\}\}$ is a perfect matching of 
$\tilde{G}$. Thus, we have a one-to-one correspondence between 
the perfect matchings of $G$ and those of $\tilde{G}$. The lemma 
is proved. 
\end{proof}

\begin{figure}
\begin{center}
\includegraphics[keepaspectratio,width=140mm]{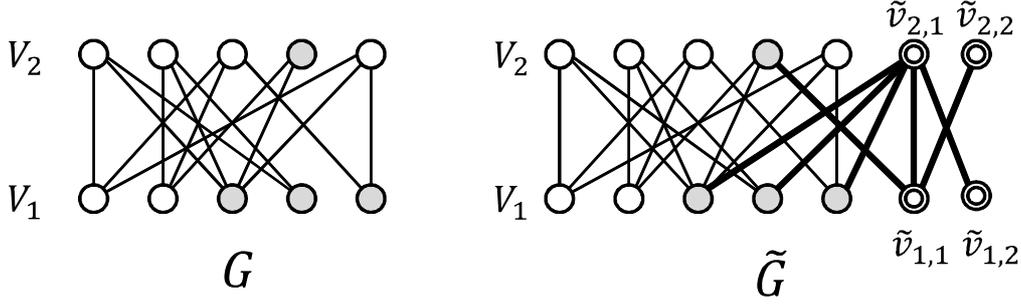}
\caption{The construction of $\tilde{G}$}
\label{fig:bipartiteConstruction}
\end{center}
\end{figure}

\section{Computing Weight Distribution}
\label{sec:Algorithm}

As seen in the previous section, the computation of the 
cut weight distribution for graph $\tilde{G}$
induces the number of perfect matchings of $G$. Thus, in what follows,
we focus on algorithms for computing the cut weight distribution.

The set of edges constituting a cut is associated with a partition
of all vertices: A partition $(S, V \setminus S)$ of all vertices $V$
induces a cutset, which is the set of edges crossing between $S$ and 
$V \setminus S$. Thus we often use the sentence ``partition 
$(S, V \setminus S)$ of $V$'' as the meaning of 
the cut associated with that partition. We define $c(S, T)$ to be 
the set of edges crossing two disjoint subsets $S$ and $T$ ($S, T \subseteq V$).
In particular, if $S$ (resp. $T$) is a singleton $\{ v\}$, we use notation 
$c(v, T)$ (resp. $c(S, v)$).

While two different partitions can lead the same cutset (e.g., 
$(S, V \setminus S)$ and $(V \setminus S, S)$), it is well-known that 
exactly $2^d$ subsets induce the same cutset, where $d$ is the number 
of connected components of $G$ and equal to $n - \mathrm{rank}(A^G)$. 
Thus, instead of computing $W_{C^{\perp}(G)}$, we rather consider the cut-weight distribution $W'_{C^{\perp}(G)}$ over all 
partitions, that is, $W'_{C^\perp(G)}[k] = |\{S \subseteq V | |c(S, V \setminus S)| = k\}|$. It is easy to calculate $W_{C^{\perp}(G)}$ from 
$W'_{C^{\perp}(G)}$ because of the relation of $W_{C^{\perp}(G)} = 
2^{-d} \cdot W'_{C^{\perp}(G)}$.

\subsection{$O^{\ast}(2^n)$-time Algorithm}

A straightforward way of computing $W'_{C^{\perp}(G)}$ is 
to enumerate all partitions of $V$ with computing their weights,
which trivially takes $O^{\ast}(2^{2n})$ time. In the case of bipartite graphs, 
we can reduce the time required for computing the cut-weight distribution.
As a first step, this subsection proposes an $O^{\ast}(2^{n})$-time algorithm,
which has the same performance as Ryser's one~\cite{Ryser63} 
(in terms of the base of the exponential part). Further improvement
of the running time is considered in the following subsection.

Let $G = (V_1 \cup V_2, E)$ be the input bipartite graph such that
$|V_1| = |V_2| = n$ and $|E| = m$, and $V = V_1 \cup V_2$ for short.
For weight-distribution vector $W$ and integer value $x \in [-m, m]$, 
we define $\sigma_x(W)$ as the vector obtained by shifting each element 
of $W$ $x$ times. That is,
\begin{equation*}
\sigma_x(W)[i] = \left\{
\begin{array}{ll}
0 & \mbox{if $i < x$},\\
W[i - x] & \mbox{if $n - 1 \geq i \geq x$}, \\
0 & \mbox{if $i \geq n + x$}.
\end{array}
\right. 
\end{equation*}

Note that the case of $i < x$ or $i \geq n + x$ applies only when
$x$ is positive or negative respectively.
Let $V'$ be a subset of $V$. We say that partition $(S, V \setminus S)$ is 
{\em conditioned} by a subset partition $(S', V' \setminus S')$ if 
$S \supseteq S'$ and $(V \setminus S) \supseteq (V' \setminus S')$ holds.
Let $\mathcal{P}_{S'|V'}$ be the set of all partitions of $V$ 
conditioned by $(S', V' \setminus S')$, and $W_{S'|V'}$ be the 
cut-weight distribution over all partitions in $\mathcal{P}_{S'|V'}$.
Our algorithm relies on the fact that $W_{S|V_1}$ can be computed within
polynomial time in $n$ provided that a partition $(S, V_1 \setminus S)$ of 
$V_1$ is given. 
In the following argument, we introduce an arbitrary ordering 
$v_0, v_1, \cdots v_{n - 1}$ of vertices in $V_2$. We define 
$V^i = \{v_i, v_{i+1}, \cdots v_{n - 1}\} \cup V_1$. The lemma 
behind the correctness of our algorithm is stated below:

\begin{lemma} \label{lma:Recursion}
For a given partition $(S, V_1 \setminus S)$, 
let $l = |c(v_i, V_1 \setminus S)| - |c(v_i, S)|$.
Then $W_{S |V^{i+1}} = W_{S | V^i} + 
\sigma_{l}(W_{S | V^i})$ holds.
\end{lemma}

\begin{proof} \normalfont
From the definition of $W_{S|V^i}$, $W_{S|V^{i+1}} = 
W_{S\cup\{v_i\}|V^i} + W_{S|V^i}$ clearly holds.
Thus it suffices to show $W_{S\cup\{v_i\}|V^i} = 
\sigma_{l}(W_{S | V^i})$. Let $(S', V \setminus S')$ 
be a partition in $\mathcal{P}_{S|V^i}$, and $k$ be its weight. 
By adding $v_i$ to $S'$, the weight increases by $l$. That is, the weight
of partition $(S' \cup \{v_i\}, V \setminus (S' \cup \{v_i\})$ is $k + l$.
It implies a one-to-one correspondence between the partitions in 
$\mathcal{P}_{S|V^i}$ with weight $k$  and those in $\mathcal{P}_{S 
\cup \{v_i\}|V^i}$ with weight $k + l$. Hence we have 
$W_{S\cup\{v_i\}|V^i}[k + l] = W_{S | V^i}[k]$ for any $k$. It clearly
follows $W_{S\cup\{v_i\}|V^i} = \sigma_{l}(W_{S | V^i})$.
The lemma is proved. \qed
\end{proof}

The recursive formula in Lemma~\ref{lma:Recursion} trivially allows
us to compute $W_{S|V_1} = W_{S|V^{n - 1}}$ within polynomial time in $n$. 
For the usefulness of the following argument, we encapsulate this
recursion process by function {\sf shift} shown in the pseudocode of 
Algorithm~\ref{algo:basic}. Let $L: 2^{V_1} \to \mathbb{Z}^{|V_1|}$ be 
the function such that
$L(X)[i] = |c(v_i, V_1 \setminus X)| - |c(v_i, X)|$ holds for any 
$v_i \in V_2$. Our $O^{\ast}(2^n)$-time algorithm computes and sums up 
the values of {\sf shift}($W_{X|V^0}, L(X)$) over all partitions of $V_1$. 
That is, our algorithm computes the right side of 
the following equality:
\begin{align}
W'_{C^{\perp}(G)} &= \sum_{S \subseteq V_1} \mathsf{shift}(W_{S|V^0}, L(S)). 
\label{math:basicAlgo}
\end{align}
The correctness of this formula is obvious from the definition of $W_{S|V_1}$.

\begin{theorem}
There is an algorithm computing $W'_{C^{\perp}(G)}$ with
$O^{\ast}(2^n)$ time.
\end{theorem}

\Newcodeline
\begin{algorithm}[t]
\caption{{\sf shift}: Function for computing $W_{X|V^{n - 1}}$}
\label{algo:basic}
\begin{tabbing}
111 \= 11 \= 11 \= 11 \= 11 \= 11 \= 11 \= \kill
\Cl \> {\bf function} {\sf shift}($W, L$) \` 
\Scomment{$W \in \mathbb{N}^m$ and
, $L \in \mathbb{N}^{\ast}$} \crm
\Cl \> \> {\bf while} $L$ is not empty {\bf do} \crm
\Cl \> \> \> $l \ot$ the head of $L$ \crm
\Cl \> \> \> Remove the head of $L$ \crm
\Cl \> \> \> $W \ot W + \sigma_{l}(W)$ \crm
\Cl \> \> {\bf endwhile} \crm
\Cl \> \> {\bf return} $W$%
\end{tabbing}%
\end{algorithm}

\subsection{Function {\sf shift} as a Linear Transformation}

Before introducing the faster algorithm, we show several properties
of Function {\sf shift}. Let $H = \{h_{i,j}\}\in \mathbb{R}^{m \times m}$ 
be the matrix defined as $h_{i,j} = 1$ if $j = i + 1$ and $0$ otherwise. 
It is easy to
check this matrix works as the operator $\sigma_1$, i.e., for any
$m$-dimensional vector $W$, $WH^x = \sigma_x(W)$ holds. Hence we can describe
function {\sf shift}($W, L$) for a given sequence 
$L = (l_0, l_1, \cdots l_{n-1})$ as follows:
\begin{align}
\mathsf{shift}(W, L) &= W \left( \prod_{i = 0}^{n-1} (H^{l_i} + I) \right),
\label{math:shiftLinear}
\end{align}
where $I$ be the $m \times m$ identity matrix. 
We can obtain the following lemma:

\begin{lemma} \label{lma:propShift}
Letting $L$ and $L'$ be two sequences of integers, and 
$W_1, W_2 \in \mathbb{N}^m$. Then the following properties hold:
\begin{enumerate}
\item 
$\sigma_x(\mathsf{shift}(W, L)) = \mathsf{shift}(\sigma_x(W), L)$,
\item $\mathsf{shift}(\mathsf{shift}(W, L), L') = \mathsf{shift}(W, L \circ L')$,
\item
$\sigma_x(W_1 + W_2) =\sigma_x(W_1) + \sigma_x(W_2)$, and 
$\mathsf{shift}(W_1 + W_2, L) = \mathsf{shift}(W_1, L) + \mathsf{shift}(W_2, L)$,

\end{enumerate}
where $\circ$ is the concatenation of two sequences.
\end{lemma}

\begin{proof} \normalfont

Since $\sigma_x(W) = \mathsf{shift}(W, (x))$, we can treat $\sigma_x$
equivalently to ${\sf shift}$. Clearly, Equation~\ref{math:shiftLinear} 
implies that {\sf shift}$(\ast, L)$ is a commutative linear 
transformation. Thus all properties obviously hold. \qed

\end{proof}

\subsection{Improving Running Time}

In this subsection, we consider an improvement of $O^{\ast}(2^n)$-time algorithm.
The running time of the improved algorithm is 
$O^{\ast}(2^{(1 - \frac{1}{5\Delta \log \Delta})n})$ and consumes 
exponential space, where $\Delta$ is the average degree of the input graph. 

The underlying principle of the improved algorithm is very simple: Separating 
two smaller subproblems. Let $(T_1, U_1)$ be a partition of $V_1$ (i.e., 
$T_1 = V_1 \setminus U_1$) fixed by the algorithm, $N(U_1) \subseteq
V_2$ be the set of vertices adjacent to $U_1$, and
$v_0, v_1, \cdots v_{n - 1}$ be an arbitrary ordering of $V_2$ such that
the last $|N(U_1)|$ vertices correspond to $N(U_1)$. The cardinality
of $N(U_1)$ is denoted by $h$ for short. Now we consider the situation 
where $U_1$ and $T_1$ are partitioned into $(X, U_1 \setminus X)$ and 
$(Y, T_1 \setminus Y)$. If we regards $X$ and $Y$ 
as variables, the first $n - h$ entries $(l_0, l_1, \cdots l_{n - h})$ 
of $L(X \cup Y)$ become a function of $X$, which are independent of the 
value of $Y$. In contrast, the last $h$ entries 
$(l_{n - h}, l_{n - h +1}, \cdots l_{n-1})$ are a function of both 
$X$ and $Y$. Consequently, by two  appropriate functions
$L_T : 2^{|T_1|} \to \mathbb{Z}^{n - h}$ and 
$L_U : 2^{T_1} \times 2^{U_1} \to \mathbb{Z}^{h}$, 
the sequence $L(X \cup Y)$ can be described as follows:
\[
L(X, Y) = L_T(X) \circ L_U(X, Y).
\]

Then the following lemma holds: 

\begin{lemma} \label{lma:sequenceSum}
$L_U(X, Y) = L_U(X, \emptyset) + L_U(\emptyset, Y) - 
L_U(\emptyset, \emptyset)$. 
\end{lemma}

\begin{proof} \normalfont
We prove $L_U(X, Y)[i] = L_U(X, \emptyset)[i] + L_U(\emptyset, Y)[i] 
- L(\emptyset, \emptyset)[i]$ for any $i$.
Since  $X \subseteq T_1$ and $Y \subseteq U_1$ are mutually disjoint, 
the sets of edges $c(v_i, X)$ and $c(v_i, Y)$ are
mutually disjoint. Thus we have $|c(v_i, X \cup Y)| = |c(v_i, X)| + |c(v_i, Y)|$.
Similarly, we have $|c(v_i, V_1 \setminus (X \cup Y))| = 
|c(v_i, (T_1 \setminus X) \cup (U_1 \setminus Y))| = |c(v_i, (T_1 \setminus X))|
+ |c(v_i, (U_1 \setminus Y))|$. Then we can obtain the following equality:
\begin{align*}
L_U(X, Y)[i] &= |c(v_i, V_1 \setminus X \cup Y)| - |c(v_i, (X \cup Y))| \\
&= |c(v_i, (T_1 \setminus X))| - |c(v_i, X)| + |c(v_i, (U_1 \setminus Y))| 
- |c(v_i, Y)|\\
&= |c(v_i, (V_1 \setminus X))| - |c(v_i, T_1)| - |c(v_i, X)| 
+ |c(v_i, (V_1 \setminus Y))| - |c(v_i, U_1)| - |c(v_i, Y)| \\
&= L_U(X, \emptyset)[i] + L_U(\emptyset, Y)[i] - |c(v_i, T_1)| - |c(v_i, U_1)| \\
&= L_U(X, \emptyset)[i] + L_U(\emptyset, Y)[i] - |c(v_i, T_1 \cup U_1)| \\
&= L_U(X, \emptyset)[i] + L_U(\emptyset, Y)[i] - L_U(\emptyset, \emptyset)[i].
\end{align*}
The lemma is proved. \qed
\end{proof}

The improved algorithm runs as follows:
\begin{itemize}
\item (Step 1) We divide all partitions of $T_1$
into several classes $\mathcal{C}_0, \mathcal{C}_1, \cdots \mathcal{C}_x$
such that for any two partitions $(X_1, T_1 \setminus X_1)$ and 
$(X_2, T_1 \setminus X_2)$ in the same class, $L_U(X_1, \emptyset) = 
L(X_2, \emptyset)$ holds.
\item (Step 2) For each $i \in [1, x]$, we compute weight
distribution $W_i = \sum_{(X, T_1 \setminus X) \in \mathcal{C}_i} 
W_{X|V^{n - h - 1}}$.

(Note that $W_i = \sum_{(X, T_1 \setminus X) \in \mathcal{C}_i} 
\mathsf{shift}(W_{X|V^0}, L_T(X))$ holds.)
\item (Step 3) Let $L(i)$ be the value of $L_U(X, \emptyset)$ 
associated with class $\mathcal{C}_i$ and $c_Y = |c(Y, V_2)|$ for short.
For each $i \in [0, x]$ and each partition $(Y, U_1 \setminus Y)$ of 
$U_1$, we compute $L_U(i, Y) = L(i) + L_U(\emptyset, Y) -
L_U(\emptyset, \emptyset)$ and 
{\sf shift}($\sigma_{c_Y}(W_i), L_U(i, Y)$). The sum of all 
the values returned by function {\sf shift} is the output of the algorithm.
\end{itemize}

We can show the following lemma, which directly leads the correctness
of the algorithm:

\begin{lemma} \label{lma:Sum}
$W'_{C^{\perp}(G)} = \sum_{i = 1}^{x} \sum_{Y \subseteq U_1} 
\mathsf{shift}(\sigma_{c_Y}(W_i), L_U(i, Y))$. 
\end{lemma}

\begin{proof} \normalfont
Since $W_{X|V^{0}}$ is the distribution over
singleton $\{(X, V^0 \setminus X)\}$, we have 
$W_{X|V^0}[i] = 1$ for $i = |c(X, V^0 \setminus X)|$ and $0$ 
otherwise. Thus, we have $\sigma_{c_Y}(W_{X|V^0})[i] = 1$
for $i = |c(X, V^0 \setminus X)| + c_Y$ and $0$ otherwise.
Since $|c(X, V^0 \setminus X)| + c_Y = |c(X \cup Y, 
V \setminus (X \cup Y))|$ holds, we obtain
\begin{align}
\sigma_{c(Y)}(W_{X|V^0}) &= W_{X \cup Y | V^0}.
\label{math:initDistribution}
\end{align}
By using this equation, Lemma~\ref{lma:propShift} and \ref{lma:Sum}, 
we can obtain the following equality:
\begin{align*}
\lefteqn{\sum_{i = 1}^{x} \sum_{Y \subseteq U_1} 
\mathsf{shift}(\sigma_{c_Y}(W_i), L_U(i, Y))} \hspace{7mm} \crm
&= \sum_{i = 1}^{x} \sum_{Y \subseteq U_1} 
\mathsf{shift}\left(\sigma_{c_Y}\left(\sum_{(X, T_1 \setminus X) 
\in \mathcal{C}_i} W_{X|V^{n - h - 1}}\right), L_U(i, Y)\right) 
%&& \text{(by the definition of $W_i$)}
\\
&= \sum_{i = 1}^{x} \sum_{Y \subseteq U_1} 
\sum_{(X, T_1 \setminus X) \in \mathcal{C}_i}
\mathsf{shift}\left( 
\sigma_{c_Y}(W_{X|V^{n - h - 1}}), L_U(X, Y)\right)
%&& \text{(by Prop. 2 of Lemma~\ref{lma:propShift})} 
\\
&= \sum_{i = 1}^{x} \sum_{Y \subseteq U_1} 
\sum_{(X, T_1 \setminus X) \in \mathcal{C}_i} 
\mathsf{shift}\left(\sigma_{c_Y}(\mathsf{shift}(W_{X|V^{0}}, L_T(X))), L_U(X, Y)\right)
%&& \text{(by Def. of $W_{X|V^{n - h - 1}}$)} 
\\
&= \sum_{i = 1}^{x} \sum_{Y \subseteq U_1} 
\sum_{(X, T_1 \setminus X) \in \mathcal{C}_i} 
\mathsf{shift}\left(
\mathsf{shift}(\sigma_{c_Y}(W_{X|V^{0}}), L_T(X)), L_U(X, Y)\right)
%&& \text{(by Prop. 2 of Lemma~\ref{lma:propShift})} 
\\
&= \sum_{i = 1}^{x} \sum_{Y \subseteq U_1} 
\sum_{(X, T_1 \setminus X) \in \mathcal{C}_i} 
\mathsf{shift}\left(\sigma_{c_Y}(W_{X|V^{0}}), L_T(X) \circ L_U(X, Y)\right)
%&& \text{(by Prop. 2 of Lemma~\ref{lma:propShift})} 
\\
&= \sum_{X \subseteq T_1, Y \subseteq U_1} 
\mathsf{shift}(W_{X \cup Y|V^0}, L(X \cup Y)) 
%&& \text{(by Prop. 1 of Lemma~\ref{lma:propShift})} 
\\
&= W'_{C^{\perp}(G)}.
%&& \text{(by Eq.~\ref{math:basicAlgo})}
\end{align*}
The lemma is proved. \qed
\end{proof}

We focus on the running time of the algorithm. Clearly the first and
second steps of the algorithm take $O^{\ast}(2^{n - |U_1|})$ time 
respectively. The third step requires time of $O^{\ast}(x2^{|U_1|})$.
Thus the total running time is $O^{\ast}(2^{n - |U_1|} + x2^{|U_1|})$.

How small can we bound $x$? Clearly, it is upper bounded 
by the size of the domain of $L_U(X)$. 
From the definition, the value of $L_U(X)[i - (n - h)]$ can
take $d(v_i) + 1$ different values  for any $v_i \in N(U_1)$. 
It follows $x \leq \prod_{v_i \in N(U_1)} (d(v_i) + 1)$.
By applying the arithmetic mean-geometric mean inequality,
we can further bound $x$ by 
$((\sum_{v_i \in N(U_1)} (d(v_i)+ 1))/|N(U_1)|)^{|N(U_1)|}$.
Letting $\Delta_X$ be the average degree over $X \subseteq V$ in $G$, 
we have 
\begin{align}
x &\leq (\Delta_{N(U_1)}+1)^{|N(U_1)|}.
\label{math:boundX} 
\end{align}
We consider how to choose $U_1$. Letting $\Delta$
be the average degree of $G$, $V_1$ contains a subset $X$ of $n/5$ 
vertices whose degrees are at most $5\Delta/4$. We choose 
$n/(5 \Delta \log \Delta)$ vertices from $X$ as $U_1$. For that choice we have
$|N(U_1)| \leq n/(4 \log \Delta)$. Since $|N(U_1)| \Delta_{N(U_1)} 
\leq \Delta n$
holds, we obtain $\Delta_{N(U_1)} \leq 4 \Delta \log \Delta$.
By assigning this bound to Inequality~\ref{math:boundX}, we obtain
\[
x \leq (4 \Delta \log \Delta + 1)^{\frac{n}{4 \log \Delta}} 
\leq (4 \Delta^{2})^{\frac{n}{4 \log \Delta}} = 
O (2^{\frac{5n}{6}})
\] 
Consequently, it follows that the running time of our algorithm is 
$O^{\ast}(2^{(1 - \frac{1}{5 \Delta \log \Delta})n})$. 

\begin{theorem}
There is an an algorithm for counting perfect matchings of bipartite graphs
which runs with $O^{\ast}(2^{(1 - \frac{1}{5 \Delta \log \Delta})n})$ time 
and exponential space.
\end{theorem}

\section{Related Work}
\label{sec:RelatedWork}

As seen in the introduction, we have roughly three lines about the studies
on counting perfect matchings. We introduce the related work along them
respectively. 

There has been proposed two different approach for approximating the
number of perfect matchings. The first one is the 
Markov-chain Monte Carlo method, which gives a fully-polynomial 
randomized approximation scheme (FPRAS) for counting 
perfect matchings~\cite{JS89,JSV04,BSVV06}. The second one is 
a randomized averaging of the determinant~\cite{GG81,KKLLL93,CRS02}. 
The fastest approximation algorithm on this approach is one
by Chien et.al.~\cite{CRS02}, which runs with $O(1.2^n)$ time. It is 
still an open problem whether there exists a FPRAS following this 
approach or not.

The second line is the algorithm design for restricted inputs. 
A seminal work on this line is a polynomial-time exact counting algorithm 
for planar graphs~\cite{Kasteleyn67}. As other restrictions, graphs of 
bounded genus~\cite{GL99,Tesler00} or bounded treewidth~\cite{ALS91,RBR09}, 
and chordal graphs with its subclass~\cite{OUU09} are considered.

About the line of exact algorithms, we have already mentioned the results
for bipartite graphs in the introduction. Thus we introduce only the work 
on counting perfect matchings for general graphs. A first result breaking
the trivial $O^{\ast}(2^m)$-time bound is one by Bj\"{o}rklund and 
Husfeldt~\cite{BH08}, which has shown two algorithms: The first one
runs with $O^{\ast}(2^{2n})$ time and polynomial space, and 
the second rounds with $O^{\ast}(1.733^{2n})$ time and exponential space. 
These algorithms are similar with our result in the sense that it also 
reduces the problem into a counting over a different universe. A number 
of the following studies improve this
bound~\cite{Koivisto09,Nederlof08,AFS09,Nederlof10,Bjorklund12}. 
The most recent and fastest one is the algorithm by 
Bj\"{o}rklund~\cite{Bjorklund12}, which achieves the same running time 
as Ryser's algorithm (that is, currently we do not find the difference of 
inherent difficulty between bipartite and general graphs). About time 
complexity,  Dell et.al.~\cite{DHW10} has shown that any algorithm 
has an instance of $m$ edges incurring $\Omega(exp(m/\log m))$ time 
if we believe that a counting version of the Exponential Time 
Hypothesis~\cite{IPZ01} is true.

\section{Concluding Remarks}
\label{sec:Conclusion}

In this paper, we presented a new algorithm for the problem of 
counting perfect matchings, which has an improved time bound 
depending on the average degree $\Delta$ of the input graph. 
Compared to previous results, our algorithm runs faster 
for many cases. In particular, the performance degradation to 
the increase of $\Delta$ is quite slower than the previous 
algorithms. The main idea of our algorithm is a new reduction 
to computing the cut-weight distribution of the input graph. 
Our algorithm is designed by combining this reduction with a 
novel algorithm for the computation of cut-weight distribution. 
The approach itself is quite new, and may be of independent interest.
Finally, we conclude the paper with several open problems related
to our approach.
\begin{itemize}
\item Can we achieve the running time exponentially faster than
Ryser's one by designing a faster algorithm computing cut-weight
distribution? 
\item The reduction part of our result is directly
applicable to any graph (which may not be bipartite). Can we 
use the reduction
to obtain a faster algorithm for general graphs? Actually,
letting $I(G)$ be the independent sets of the input graph $G$, we 
can easily obtain an algorithm with $O^{\ast}(2^{2n - |I(G)|})$ running time
by regarding $G$ as a "quasi" bipartite graph of two vertex sets $I(G)$ and $V \setminus I(G)$ and applying our $O^{\ast}(2^n)$-time algorithm, which gives
the same performance as the algorithm by~\cite{AFS09}. 
\item Is it possible to design a faster FPRAS for counting perfect
matchings based on our method? Note that an 
$(1 + \epsilon)$-approximation of the cut-weight distribution
trivially induces an $(1 + \epsilon)$-approximation of the
number of perfect matchings because of the linearity of the transformation.  
\item Computing cut-weight distribution is a special case of the counting
version of 2-CSP, which is addressed by Williams~\cite{Williams05}. 
In this sense, our reduction gives a new linkage from counting perfect 
matchings to CSP. 
Can we use this linkage for obtaining some new complexity result around
those problems?
\item Can we apply the same technique to other combinatorial 
problems? Interestingly, there has been proposed a variety of
MacWilliams-style Identities in the field of the coding theory.
We may find a useful transformation from those resources. In 
addition, it may be an interesting approach to focus on the 
primal-dual relationship of two universes. Can we design
a kind of primal-dual algorithms for counting problems?
\end{itemize}

\bibliographystyle{plain}
\bibliography{books}

\end{document}